\newif\iflong\longtrue

\documentclass{article}
\usepackage{ijcai21}

\usepackage{times}

\usepackage{soul}
\usepackage{url}
\usepackage[hidelinks]{hyperref}
\usepackage[utf8]{inputenc}
\usepackage[small]{caption}
\usepackage{graphicx}
\usepackage{amsmath}
\usepackage{booktabs}
\usepackage[english]{babel}
\usepackage[utf8]{inputenc}
\usepackage[T1]{fontenc}
\usepackage{amsmath,amssymb}
\usepackage{amsthm}
\usepackage{todonotes}
\usepackage{graphicx}
\usepackage[noend]{algpseudocode}
\usepackage[hidelinks]{hyperref}
\usepackage{dsfont}
\usepackage{braket}
\usepackage[shortlabels]{enumitem}
\usepackage{tikz}
\usepackage{verbatim}
\usepackage{todonotes}
\usepackage{xspace}
\usepackage{cleveref}
\usepackage{nicefrac}
\usepackage[linesnumbered,vlined,commentsnumbered, ruled]{algorithm2e}

\newcommand{\W}[1]{\ensuremath{\mathrm{W}[#1]}\xspace}
\newcommand\NP{\ensuremath{\mathrm{NP}}\xspace}
\newcommand\coNP{\ensuremath{\mathrm{coNP}}\xspace}

\newcommand{\IS}{\textsc{Independent Set}\xspace}
\newcommand{\PL}{\textsc{Polytree Learning}\xspace}

\newcommand{\MCIS}{\textsc{Multicolored Independent Set}\xspace}

\newcommand{\prob}[3]{\begin{center}
	\begin{minipage}[c]{.9\linewidth}
          \textsc{#1}\\
          \textbf{Input}: #2\\
          \textbf{Question}: #3
	\end{minipage}
\end{center}}

\newtheorem{theorem}{Theorem}[section]
\newtheorem{lemma}[theorem]{Lemma}
\newtheorem{definition}[theorem]{Definition}
\newtheorem{proposition}[theorem]{Proposition}
\newtheorem{corollary}[theorem]{Corollary}

\newtheorem{claim}{Claim}
\theoremstyle{definition}

{\bfseries}{\itshape}
{\bfseries}{\itshape}
\theoremstyle{definition}
\newtheorem*{claimproof}{\normalfont{\textit{Proof}}}

\definecolor{myred}{rgb}{1,0.25,0.25}

\newcommand{\Oh}{\mathcal{O}} 
\newcommand{\Fa}{\mathcal{F}} 
\newcommand{\Ind}{\mathcal{I}} 

\newcommand{\score}{\text{\normalfont{score}}}
\newcommand{\poly}{\text{\normalfont{poly}}}

\newcommand{\bsh}{\ensuremath{\NP \subseteq\coNP / \poly}}

\title{On the Parameterized Complexity of Polytree~Learning}

\author{
Niels Grüttemeier\and
Christian Komusiewicz\And
Nils Morawietz\thanks{Supported by the Deutsche Forschungsgemeinschaft (DFG), project OPERAH, KO~3669/5-1.}
\affiliations
Philipps-Universität Marburg, Marburg, Germany\\
\emails
\{niegru, komusiewicz, morawietz\}@informatik.uni-marburg.de
}

\begin{document}
\maketitle
\begin{abstract}
A Bayesian network is a directed acyclic graph that represents statistical dependencies between variables of a joint probability distribution. A fundamental task in data science is to learn a Bayesian network from observed data. \textsc{Polytree Learning} is the problem of learning an optimal Bayesian network that fulfills the additional property that its underlying undirected graph is a forest. In this work, we revisit the complexity of \textsc{Polytree Learning}. We show that \textsc{Polytree Learning} can be solved in~$3^n \cdot |I|^{\Oh(1)}$ time where $n$~is the number of variables and~$|I|$ is the total instance size.  Moreover, we consider the influence of the number of variables~$d$ that might receive a nonempty parent set in the final DAG on the complexity of \textsc{Polytree Learning}. We show that \textsc{Polytree Learning} has no~$f(d)\cdot |I|^{\Oh(1)}$-time algorithm, unlike Bayesian network learning which can be solved in~$2^d \cdot |I|^{\Oh(1)}$~time. We show that, in contrast, if~$d$ and the maximum parent set size are bounded, then we can obtain efficient algorithms. 
  \end{abstract}

  \section{Introduction}
  Bayesian networks are the most important tool for modelling statistical dependencies in
  joint probability distributions. A Bayesian network consists of a DAG~$(N,A)$
  over the variable set~$N$ and a set of condiditional probability tables. Given a Bayesian
  network and the observed values on some of its variables, one may infer the probability
  distributions of the remaining variables under the observations. One of the drawbacks of
  using Bayesian networks is that this inference task is NP-hard. Moreover, the task of
  constructing a Bayesian network with an optimal network structure is NP-hard as well,
  even on very restricted instances~\cite{C95}. In this problem, we are given a local parent score function~$f_v: 2^{N \setminus \{v\}} \to \mathds{N}_0$ for each variable~$v$ and the task is to learn a DAG~$(N,A)$ such that the sum of the parent scores over all variables is maximal.

  In light of the hardness of handling general
  Bayesian networks, the learning and inference problems for Bayesian networks fulfilling
  some specific structural constraints have been studied
  extensively~\cite{Pearl89,KP13,KP15,GK20,VS21}. 

  One of the earliest special cases that has received attention are \emph{tree} networks,
  also called \emph{branchings}. A tree is a Bayesian network where every vertex has at
  most one parent. In other words, every connected component is a directed
  in-tree. Learning and inference can be performed in polynomial time on
  trees~\cite{CL68,Pearl89}. Trees are, however, very limited in their modeling power
  since every variable may depend only on at most one other variable. To overcome
  this problem, a generalization of branchings called polytrees has been
  proposed. A polytree is a DAG whose underlying undirected graph is a forest.
  An advantage of polytrees is that the inference task can be performed in polynomial time
  on them~\cite{Pearl89,GH02}. \textsc{Polytree
    Learning}, the problem of learning an optimal polytree structure from
  parent scores, however, remains NP-hard~\cite{D99}.  We study exact
  algorithms for~\textsc{Polytree
    Learning}.

  \paragraph{Related Work.}\textsc{Polytree Learning} is NP-hard even if every 
  parent set with strictly positive score has size at most~2~\cite{D99}. Motivated by the contrast between the
  NP-hardness of \textsc{Polytree Learning} and the fact that learning a tree has a
  polynomial-time algorithm, the problem of optimally learning polytrees that are close to
  trees has been considered. More precisely, it has been shown that the best
  polytree among those that can be transformed into a tree by at most $k$ edge deletions can be
  found in~$n^{\Oh(k)}|I|^{\Oh(1)}$~time~\cite{GKLOS15,SMS13} where~$n$ is the number of variables and~$|I|$ is the overall input size. Thus, the running time of these algorithms is polynomial for every fixed~$k$. As noted by Gaspers et
  al.~\shortcite{GKLOS15}, a brute-force algorithm for \textsc{Polytree Learning} would need to
  consider $n^{n-2}\cdot 2^{n-1}$~directed trees. Polytrees have been used, for example,
  in image-segmentation for microscopy data~\cite{FGLMJF19}.
  
  \paragraph{Our Results.} We obtain an algorithm that solves \textsc{Polytree Learning}
  in~$3^n\cdot |I|^{\Oh(1)}$ time. This \iflong is the first algorithm for \textsc{Polytree
    Learning} where the running time is singly-exponential in the number of
  variables~$n$. The\fi  running time is a substantial improvement over the brute-force
  algorithm mentioned above, thus  positively answering a question of Gaspers et al.~\shortcite{GKLOS15} on
  the existence of such algorithms. We then study whether \textsc{Polytree Learning} is
  amenable to polynomial-time data reduction that shrinks the instance~$I$ if~$|I|$ is
  much bigger than~$n$. \iflong Using tools from parameterized complexity theory, we \else We \fi show that
  such a data reduction algorithm is unlikely. More precisely, we show that (under
  standard complexity-theoretic assumptions) there is no polynomial-time algorithm that
  replaces any $n$-variable instance~$I$ by an equivalent one of size~$n^{\Oh(1)}$. In
  other words, it seems necessary to keep an exponential number of parent sets to compute
  the optimal polytree.

  We then consider a parameter that is potentially much smaller than~$n$ and determine
  whether \textsc{Polytree Learning} can be solved efficiently when this parameter is small. The parameter~$d$, which we call the number of \emph{dependent} variables, is the number of
  variables~$v$ for which at least one entry of~$f_v$ is strictly positive. The parameter
  essentially counts how many variables might receive a nonempty parent set in an optimal
  solution. We show that \textsc{Polytree Learning} can be solved in polynomial time
  when~$d$ is constant but that an algorithm with running time~$g(d)\cdot |I|^{\Oh(1)}$ is
  unlikely for any computable function~$g$. Consequently, in order to obtain positive results for the
  parameter~$d$, one needs to consider further restrictions on the structure of the input
  instance. We make a first step in this direction and consider the case where all parent sets with a strictly positive score have size at most~$p$. Using
  this parameterization, we show that every input instance can be solved
  in~$2^{\omega dp} \cdot |I|^{\Oh(1)}$~time where~$\omega$ is the matrix
  multiplication constant. With the current-best known value for~$\omega$ this gives a
  running time of~$5.18^{dp}\cdot |I|^{\Oh(1)}.$ We then consider again 
  data reduction approaches. This time we obtain a positive result: Any instance of \textsc{Polytree
    Learning} where~$p$ is constant can be reduced in polynomial time to an equivalent one
  of size~$d^{\Oh(1)}$. Informally, this means that if the instance has only few dependent
  variables, the parent sets with strictly positive score are small, and there are many
  nondependent variables, then we can identify some nondependent variables that are irrelevant for
  an optimal polytree representing the input data. We note that this result is tight in
  the following sense: Under standard complexity-theoretic assumptions it is impossible to
  replace each input instance in polynomial time by an equivalent one with~$(d+p)^{\Oh(1)}$
    variables. Thus, the assumption that~$p$ is a constant is necessary.

\section{Preliminaries}

\paragraph{Notation.} 
An \emph{undirected graph} is a tuple~$(V,E)$, where~$V$ is a set of vertices and~$E \subseteq \{ \{u,v\} \mid u,v \in V\}$ is a set of edges. Given a vertex~$v \in V$, we define~$N_G(v):=\{u \in V \mid \{u,v\} \in E\}$ as the~\emph{neighborhood of~$v$}. A \emph{directed graph} is a tuple~$(N,A)$, where~$N$ is a set of vertices and~$A \subseteq N \times N$ is a set of arcs. If~$(u,v) \in A$ we call~$u$ a \emph{parent of~$v$} and~$v$ a \emph{child of~$u$}. \iflong A vertex that has no parents is a~\emph{source} and a vertex that has no children is a~\emph{sink}.\fi Given a vertex~$v$, we let~$P^A_v:=\{u \mid (u,v) \in A\}$ denote the~\emph{parent set} of~$v$. The~\emph{skeleton} of a directed graph~$(N,A)$ is the undirected graph~$(N,E)$ where~$\{u,v\} \in E$ if and only if~$(u,v) \in A$ or~$(v,u) \in A$. A directed acyclic graph is a~\emph{polytree} if its skeleton is a forest, that is, the skeleton is acyclic~\cite{D99}. As a shorthand, we write~$P \times v := P \times \{v\}$ for a vertex~$v$ and a set~$P \subseteq N$.

\paragraph{Problem Definition.} Given a vertex set~$N$, a family~$\Fa:= \{ f_v: 2^{N \setminus \{v\}} \rightarrow \mathds{N}_0 \mid v \in N\}$ is a \emph{family of local scores for~$N$}. Intuitively, $f_v(P)$ is the score that a vertex~$v$ obtains if~$P$ is its parent set. 
Given a directed graph~$D:=(N,A)$ we define~$\score(A):= \sum_{v \in N} f_v(P^A_v)$. We study the following computational problem.

\begin{center}
	\begin{minipage}[c]{.9\linewidth}
          \PL\\
          \textbf{Input}: A set of vertices~$N$, local scores~$\Fa=\{f_v \mid v \in N\}$, and an integer~$t \in \mathds{N}_0$.\\
          \textbf{Question}: Is there an arc-set~$A \subseteq N \times N$ such that~$(N,A)$ is a polytree and~$\score(A) \geq t$?
	\end{minipage}
\end{center}

Given an instance~$I:=(N,\Fa,t)$ of \textsc{Polytree Learning}, an arc-set~$A$ is a \emph{solution of~$I$} if~$(N,A)$ is a polytree and~${\score(A) \geq t}$. Without loss of generality we may assume that~$f_v(\emptyset)=0$ for every~$v \in N$~\cite{GK20}. \iflong Throughout this work, we let~$n:=|N|$. \fi

\paragraph{Solution Structure and Input Representation.}

We assume that the local scores~$\Fa$ are given in \emph{non-zero representation}. That is,~$f_v(P)$ is part of the input if it is different from~0. For~$N=\{v_1, \dots, v_n\}$, the local scores~$\Fa$ are represented by a two-dimensional array~$[Q_1, Q_2, \dots, Q_n]$, where each~$Q_i$ is an array containing all triples~$(f_{v_i}(P),|P|,P)$ where~$f_{v_i}(P)>0$. The size~$|\Fa|$ is defined as the number of bits needed to store this two-dimensional array. Given an instance~$I:=(N,\Fa,t)$, we define~$|I|:=n+|\Fa|+\log(t)$.

For a vertex~$v$, we call~$\mathcal{P}_\Fa(v):=\{P \subseteq N \setminus \{v\} \mid f_v(P)>0\} \cup \{\emptyset\}$ the~\emph{set of potential parents of~$v$}. Given a yes-instance~$I:=(N,\Fa,t)$ of~\PL, there exists a solution~$A$ such that~$P^A_v \in \mathcal{P}_\Fa(v)$ for every~$v \in N$: If there is a vertex with~$P^A_v \not \in \mathcal{P}_\Fa(v)$ we can simply replace its parent set by~$\emptyset$. The running times presentend in this paper will also be measured in the maximum number of potential parent sets~$\delta_\Fa := \max_{v \in N} |\mathcal{P}_\Fa(v)|$~\cite{OS13}. 

A tool for designing algorithms for~\PL is the \emph{directed superstructure}~\cite{OS13} which is the directed graph~$S_\Fa:=(N,A_\Fa)$ with~$A_\Fa:=\{(u,v) \mid \exists P \in \mathcal{P}_\Fa(v): u \in P\}$. In other words, there is an arc~$(u,v) \in A_\Fa$ if and only if~$u$ is a potential parent of~$v$.

\paragraph{Parameterized Complexity.} A problem is in the class XP for a parameter~$k$ if it can be solved in~$|I|^{g(k)}$~time for some computable function~$g$. In other words, a problem is in~XP if it can be solved within polynomial time for every fixed~$k$. A problem is fixed-parameter tractable (FPT) for a parameter~$k$ if it can be solved in~$g(k) \cdot |I|^{\Oh(1)}$ time for some computable~$g$. If a problem is W[1]-hard for a parameter~$k$, then it is assumed to not be fixed-parameter tractable for~$k$. A \emph{problem kernel} is an algorithm that, given an instance~$I$ with parameter~$k$ computes in polynomial time an equivalent instance~$I'$ with parameter~$k'$ such that~$|I'| + k' \leq h(k)$ for some computable function~$h$. If~$h$ is a polynomial, then the problem admits a \emph{polynomial kernel}. Strong conditional running time bounds can also be obtained by assuming the \emph{Exponential Time Hypothesis (ETH)}, a standard conjecture in complexity theory~\cite{IPZ01}. For a detailed introduction into parameterized complexity we refer to the standard textbook~\cite{CFKLMPPS15}.

\section{Parameterization by the Number of Vertices}\label{sec:number verts}
In this section we study the complexity of \textsc{Polytree Learning} when parameterized by~$n$, the number of vertices. Note that there are up to~$n\cdot 2^{n-1}$ entries in~$\Fa$ and thus, the total input size of an instance of \textsc{Polytree Learning} might be exponential in~$n$.\iflong On the positive side, we show that the problem can be solved in~$\Oh(3^n \cdot |I|^{\Oh(1)})$~time.

On the negative side, we complement the result above by proving that \PL does not admit a polynomial kernel when parameterized by~$n$. 
In other words, it is presumably impossible to transform an instance of~\PL in polynomial time into an equivalent instance of size~$|I|=n^{\Oh(1)}$. 
This shows that it is sometimes necessary to keep an exponential number of parent scores to compute an optimal polytree.\fi

\iflong

\subsection{An FPT Algorithm}
\fi

\begin{theorem}\label{the:algo for verts}
\PL can be solved in~$3^{n} \cdot |I|^{\Oh(1)}$ time.
\end{theorem}
\begin{proof}

Let~$I:=(N, \Fa, t)$ be an instance of~\PL.
We describe a dynamic programming algorithm to solve~$I$.
We suppose an arbitrary fixed total ordering on the vertices of~$N$.
For every~$P \subseteq N$, and every~$i\in [1, |P|]$, we denote with~$p_i$ the~$i$th smallest element of~$P$ according to the total ordering.

 The dynamic programming table~$T$ has entries of type~$T[v, P,S,i]$ with~$v\in N$,~$P \in \mathcal{P}_\Fa(v)$,~$S \subseteq N \setminus (P \cup \{v\})$, and~$i\in [0,|P|]$. 
 
 Each entry stores the maximal score of an arc set~$A$ of a polytree on~$\{v\}\cup P \cup S$ where~$v$ has no children,~$v$ learns exactly the parent set~$P$ under~$A$, and for each~$j\in [i+1, |P|]$, only the arc~$(p_j, v)$ is incident with~$p_j$ under~$A$. See Figure~\ref{fig:table} for an illustration.
 
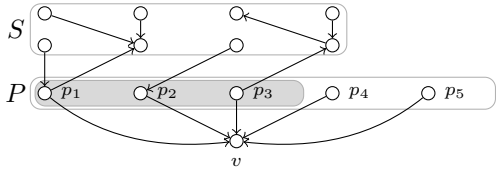
\begin{figure}

\begin{center}
\begin{tikzpicture}[scale=0.85,yscale=0.5,xscale=0.75]
\tikzstyle{knoten}=[circle,fill=white,draw=black,minimum size=5pt,inner sep=0pt]
\tikzstyle{bez}=[inner sep=0pt]

\draw[rounded corners, black!30] (0.7, .5) rectangle (10.4, -0.5) {};
\draw[rounded corners, black!30, fill=gray!30] (0.8, .4) rectangle (6.4, -0.4) {};
\node[bez] at (.4,0) {$P$};
\node[knoten,label=right:\smaller{$p_1$}] (p1) at (1,0) {};
\node[knoten,label=right:\smaller{$p_2$}] (p2) at (3,0) {};
\node[knoten,label=right:\smaller{$p_3$}] (p3) at (5,0) {};
\node[knoten,label=right:\smaller{$p_4$}] (p4) at (7,0) {};
\node[knoten,label=right:\smaller{$p_5$}] (p5) at (9,0) {};

\node[knoten,label=below:\smaller{$v$}] (v) at (5,-1.5) {};

\begin{scope}[yshift=1.5cm,xshift=-1cm]

\node[bez] at (1.4,.5) {$S$};
\draw[rounded corners, black!30] (1.7, 1.3) rectangle (8.3, -0.3) {};

\node[knoten] (u1) at (2,0) {};
\node[knoten] (u2) at (4,0) {};
\node[knoten] (u3) at (6,0) {};
\node[knoten] (u4) at (8,0) {};

\node[knoten] (o1) at (2,1) {};
\node[knoten] (o2) at (4,1) {};
\node[knoten] (o3) at (6,1) {};
\node[knoten] (o4) at (8,1) {};

\end{scope}

\draw[->]  (o1) to (u2);
\draw[->]  (o2) to (u2);
\draw[->]  (u4) to (o3);
\draw[->]  (o4) to (u4);

\draw[->]  (u1) to (p1);
\draw[->]  (p1) to (u2);
\draw[->]  (u3) to (p2);
\draw[->]  (p3) to (u4);

\draw[->, bend right=30] (p1) to (v);
\draw[->, bend right=0]  (p2) to (v);
\draw[->]  (p3) to (v);
\draw[->, bend right=-0]  (p4) to (v);
\draw[->, bend right=-30]  (p5) to (v);

\end{tikzpicture}
\end{center}
\caption{Illustration of an entry~$T[v,P,S,i]$ where~$i = 3$.}
\label{fig:table}
\end{figure}

 We initialize the table~$T$ by setting~$T[v,P,\emptyset,i]:= f_v(P)$ for all~$v\in N$,~$P\in \mathcal{P}_\Fa(v)$, and~$i\in[0,|P|]$. 
 The recurrence to compute an entry for~$v\in N$,~$P\in \mathcal{P}_\Fa(v)$,~$S \subseteq N \setminus (P \cup \{v\})$, and~$i\in [1,|P|]$ is
\begin{equation*}
\resizebox{.91\linewidth}{!}{$
    \displaystyle
    \begin{split}
    &T[v, P, S,i] := \max_{S' \subseteq S} T[v, P, S \setminus S', i-1] + \\
&~~\max_{v' \in S' \cup \{p_i\}} \max_{\substack{P' \in \mathcal{P}_\Fa(v')\\P' \subseteq S' \cup \{p_i\}}} T[v', P', (S' \cup \{p_i\}) \setminus (P' \cup \{v'\}), |P'|].
\end{split}
$}
\end{equation*}%
Note that the two vertex sets~$P \cup (S\setminus S') \cup \{v\}$ and~$P' \cup (S' \cup \{p_i\}) \setminus (P' \cup \{v'\}) \cup \{v'\} = S' \cup \{p_i\}$ share only the vertex~$p_i$. 
Hence, combining the polytree on these two vertex sets results in a polytree.
 
If~$i$ is equal to zero, then the recurrence is
\begin{equation*}
\resizebox{.91\linewidth}{!}{$
    \displaystyle
    \begin{split}
    &T[v,P,S,0] :=\\
    &~~f_v(P) + \max_{v' \in S} \max_{\substack{P' \in \mathcal{P}_\Fa(v')\\P' \subseteq S}} T[v', P', S' \setminus (P' \cup \{v'\}), |P'|].
\end{split}
$}
\end{equation*}%
Thus, to determine if~$I$ is a yes-instance of~\PL, it remains to check if~$T[v, P, N \setminus (P\cup \{v\}), |P|] \geq t$ for some~$v\in N$ and some~$P \in \mathcal{P}_\Fa(v)$.
The corresponding polytree can be found via traceback.
The correctness proof is straightforward and thus omitted.

The \iflong dynamic programming \fi table~$T$ has~$2^n \cdot n^2 \cdot \delta_\Fa$ entries. 
Each of these entries can be computed in~$2^{|S|} \cdot n^2\cdot \delta_\Fa$. 
Consequently, all entries can be computed in~$\sum_{i= 0}^n \binom{n}{i} 2^i \cdot |I|^{\Oh(1)}  = 3^n \cdot |I|^{\Oh(1)}$ time in total.
To evaluate if there is some~$v\in N$ and some~$P \in \mathcal{P}_\Fa(v)$ such that~$T[v, P, N \setminus (P\cup \{v\}), |P|] \geq t$ can afterwards be done in~$\Oh(n\cdot \delta_\Fa)$ time. 
Hence, the total running time is~$3^n \cdot |I|^{\Oh(1)}$.
\end{proof}

\iflong \paragraph{A Kernel Lower Bound} 
In this subsection we prove the following kernel lower bound.
\else
We also obtain a kernel lower bound.
\fi The proof  is closely related to a similar kernel lower bound for~\textsc{Bayesian Network Structure Learning} parameterized by~$n$~\cite{GK20}.

\begin{theorem} \label{Theorem: No-Poly for n}
\PL does not admit a polynomial kernel when parameterized by~$n$, unless~$\bsh$.
\end{theorem}

\newcommand{\proofNoPolyN}{
\begin{proof}
We give a polynomial parameter transformation~\cite{BTY11} from~\MCIS.

\prob{\MCIS}{An integer~$k$, a graph~$G=(V,E)$, and a partition~$(V_1, \dots, V_k)$ of~$V$ such that~$G[V_i]$ is a clique for each~$i\in[1,k]$.}{Does~$G$ contain an independent set of size~$k$?}
The parameter for this polynomial parameter transformation is the number of vertices not contained in~$V_k$.

Given an instance~$I=(G=(V,E), k)$ of~\MCIS with partition~$(V_1, \dots, V_k)$ with~$\ell := |V| - |V_k|$, we construct in polynomial time an equivalent instance~$I'=(N, \Fa, t)$ of~\PL such that~$|N| \in \Oh(\ell^2)$ as follows.
Let~$V^{<k} := V \setminus V_k$ and let~$E^{<k} := \{e\in E \mid e \cap V_k = \emptyset\}$. 
We set~$N := V^{<k} \cup \{w_e\mid e\in E^{<k}\} \cup \{v^*\}$.
For each~$v\in V^{<k}$, we set~$P_v := \{v^*\} \cup \{w_{\{v, u\}}\mid u \in N_G(v) \cap V^{<k}\}$ and~$f_v(P_v) := 1$.
Moreover, for each~$v\in V_k$, we set~$f_{v^*}(N_G(v) \cap V^{<k}) := 1$.
All other local scores are set to zero.
Finally, we set~$t := k$. An example of the construction is shown in Figure~\ref{Figure: NoPoly}. We show that~$I$ is a yes-instance of~\MCIS if and only if~$I'$ is a yes-instance of~\PL.

\begin{figure}
\begin{center}
\begin{tikzpicture}
\tikzstyle{knoten}=[circle,fill=white,draw=black,minimum size=5pt,inner sep=0pt]
\tikzstyle{bez}=[inner sep=0pt]

\draw[rounded corners, black!30] (-0.3, 0.3) rectangle (1.3, -0.4) {};
\node[bez] at (-0.6,0) {$V_1$};
\node[knoten] (v11) at (0,0) {};
\node[knoten] (v12) at (0.5,0) {};
\node[knoten,label=below:\smaller{$v_1$}] (v13) at (1,0) {};

\draw[rounded corners, black!30] (2.2, 0.3) rectangle (3.8, -0.4) {};
\node[bez] at (1.9,0) {$V_2$};
\node[knoten] (v21) at (2.5,0) {};
\node[knoten] (v22) at (3,0) {};
\node[knoten] (v23) at (3.5,0) {};
\node[bez] at (3.65,-0.27) {\smaller{$v_2$}};

\draw[rounded corners, black!30] (4.7, 0.3) rectangle (6.3, -0.4) {};
\node[bez] at (6.6,0) {$V_3$};
\node[knoten,label=below:\smaller{$v_3$}] (v31) at (5,0) {};
\node[knoten] (v32) at (5.5,0) {};
\node[knoten] (v33) at (6,0) {};

\node[knoten,label=below:\smaller{$v^*$}] (v*) at (3,-2) {};

\begin{scope}[yshift=1cm]
\node[knoten] (e1) at (-0.5,0) {};

\node[knoten] (e11) at (0,0) {};
\node[knoten] (e12) at (0.5,0) {};
\node[knoten] (e13) at (1,0) {};

\node[knoten] (e2) at (1.5,0) {};
\node[knoten] (e3) at (2,0) {};

\node[knoten] (e21) at (2.5,0) {};
\node[knoten] (e22) at (3,0) {};
\node[knoten] (e23) at (3.5,0) {};

\node[knoten] (e4) at (4,0) {};
\node[knoten] (e5) at (4.5,0) {};

\node[knoten] (e31) at (5,0) {};
\node[knoten] (e32) at (5.5,0) {};
\node[knoten] (e33) at (6,0) {};

\node[knoten] (e6) at (6.5,0) {};
\end{scope}

\draw[->, bend left=10]  (v*) to (v13);
\draw[->, bend right=10]  (e11) to (v13);
\draw[->]  (e12) to (v13);
\draw[->]  (e13) to (v13);
\draw[->, bend right=15]  (e1) to (v13);
\draw[->, bend left=10]  (e3) to (v13);

\draw[->]  (v*) to (v23);
\draw[->, bend right=10]  (e21) to (v23);
\draw[->]  (e22) to (v23);
\draw[->]  (e23) to (v23);
\draw[->, bend left=10]  (e5) to (v23);

\draw[->, bend right=10]  (v*) to (v31);
\draw[->]  (e31) to (v31);
\draw[->]  (e32) to (v31);
\draw[->, bend left=15]  (e33) to (v31);
\draw[->, bend right=10]  (e4) to (v31);

\draw[->, bend right]  (v11) to (v*);
\draw[->]  (v21) to (v*);
\draw[->, bend left]  (v33) to (v*);
\end{tikzpicture}
\end{center}
\caption{An example of the construction given in the proof of Theorem \ref{Theorem: No-Poly for n}. The original instance contains a multicolored independent set on the vertices~$v_1 \in V_1$, $v_2 \in V_2$, $v_3 \in V_3$, and~$v_4 \in V_4$. The directed edges represent the arcs of a polytree with score~$4$. The choice of vertex~$v_4 \in V_4$ is encoded in the parent set of~$v^*$.}\label{Figure: NoPoly}
\end{figure}
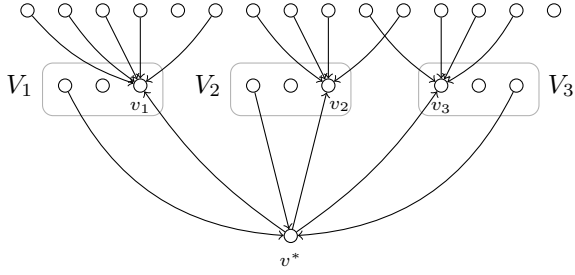

$(\Rightarrow)$
Let~$S$ be an independent set of size~$k$ in~$G$.
Hence, for each~$i\in[1,k]$, there is exactly one~$v_i\in S \cap V_i$, since~$G[V_i]$ is is a clique.
We set~$P_i := \{v^*\} \cup \{e_{\{v_i, u\}} \mid u \in N_G(v_i)\cap V^{<k}\} = P_{v_i}$ for each~$i\in [1,k-1]$ and~$P_k := N_G(v_k) \cap V^{<k}$.
By construction,~$f_{v_i}(P_i) = 1$ for each~$i\in[1,k]$ and, thus,~$D:=(N, A)$ with~$A:= \cup_{i = 1}^{k-1} P_i \times v_i \cup P_k \times v^*$ has score~$t=k$.
Hence, it remains to show that~$D$ is a polytree.
To this end, note that~$v_i \notin N_G(v_k)$ for any~$i\in[1,k-1]$ and, thus,~$v_i \notin P_k$.
Consequently,~$P_k$ is disjoint to all the sets~$\{v_i\} \cup P_i$.
Next, we show that~$(\{v_i\} \cup P_i) \cap (\{v_j\} \cup P_j) = \{v^*\}$ for distinct~$i,j\in[1,k-1],$ which implies that the skeleton of~$D$ is acyclic. 
Since~$S$ is an independent set, there is no edge~$e=\{v_i, v_j\}\in E^{<k}$ for distinct~$i,j\in[1,k-1]$.
Thus, there is no~$w_e\in P_i \cap P_j$.
Hence,~$(\{v_i\} \cup P_i) \cap (\{v_j\} \cup P_j) = \{v^*\}$ and, thus,~$D$ is a polytree.
Consequently,~$I'$ is a yes-instance of~\PL.

$(\Leftarrow)$
Let~$A\subseteq N \times N$ be an arc set such that~$D:=(N,A)$ is a polytree with score at least~$t = k$, such that each vertex~$v\in N$ learns a potential parent set~$P^A_v \in P_\Fa(v)$ under~$A$.
By construction and since~$D$ has score at least~$t$, there is a set~$S'$ of at least~$k-1$ vertices of~$V^{<k}$, such that~$v$ learns a parent set of score one under~$A$.
Since~$G[V_i]$ is a clique,~$S'$ contains exactly one vertex~$v_i\in V_i$ for each~$i\in[1,k-1]$ as, otherwise, the skeleton of~$D$ contains the cycle~$(v^*, v_i, w_{\{v_i, v_i'\}}, v_i')$ for a vertex~$v_i'\in S'\cap V_i$.
Moreover, there is some~$v_k\in V_k$ such that~$v^*$ learns the parent set~$P_k := N_G(v_k) \cap V^{<k}$ under~$A$.
We set~$S := S' \cup \{v_k\}$ and show that~$S$ is an independent set in~$G$.
 
Let~$P_i := \{v^*\} \cup \{w_{\{v_i, u\}}\mid u\in N_G(v_i) \cap V^{<k}\}$ be the parent set of~$v_i$ under~$A$ for each~$i\in[1,k-1]$.
Since~$D$ is a polytree and~$v^*\in P_i$ for each~$i\in[1,k-1]$, it follows that~$S' \cap P_k = S' \cap N_G(v_k) = \emptyset$ as, otherwise,~$(v^*, v')$ is a cycle in the skeleton of~$D$ for any~$v'\in N_G(v_k) \cap S'$.
Moreover, there is no edge~$\{v_i, v_j\}\in E$ for distinct~$i,j\in[1,k-1]$ as, otherwise,~$(v_i, w_{\{v_i, v_j\}}, v_j, x)$ is a cycle in the skeleton of~$D$.
Consequently,~$S$ is an independent set of size~$k$ in~$G$ and, thus,~$I$ is a yes-instance of~\MCIS.

Unless~$\bsh$,~\MCIS does not admit a polynomial kernel when parameterized by~$\ell$~\cite{GK20STC}.
Since~$|N| \leq \binom{\ell}{2} + \ell + 2$ it follows that~\PL does not admit a polynomial kernel when parameterized by~$|N|$, unless~$\bsh$.
\end{proof}
}
\iflong
\proofNoPolyN
\fi



\section{Dependent Vertices}\label{sec:hungry verts}
We now introduce a new parameter~$d$ called~\emph{number of dependent vertices}. Given an instance~$(N,\Fa,t)$ of~\PL, a vertex~$v \in N$ is called \emph{dependent} if there is a nonempty potential parent-set~$P \in \mathcal{P}_\Fa(v)$. Thus, a vertex is dependent if it might learn a nonempty parent set in a solution. A vertex that is not dependent is called~\emph{nondependent vertex}.  Observe that~$d$ is potentially smaller than~$n$. We start with a simple XP-result.

\begin{theorem}\label{thm:delta-f-d}
\textsc{Polytree Learning} can be solved in~${(\delta_\Fa)}^d \cdot n^{\Oh(1)}$~time.
\end{theorem}

\begin{proof}
Choose for each dependent vertex~$v_i$ one of its potential parent sets~$P_i \in \mathcal{P}_\Fa(v_i)$ and check afterwards if~$(N, \cup_{i = 1}^d P_i\times v_i)$ is a polytree of score at least~$t$.
This is the case for some combination of potential parent sets if and only if the instance is a yes-instance.
Since each check can be done in polynomial time and there are~${(\delta_\Fa)}^d$ many combinations of potential parent sets, we obtain the stated running~time.  
\end{proof}

We next show that there is little hope for a significant running~time improvement on this simple brute-force algorithm. More precisely, we show that there is no~$g(d)\cdot |I|^{\Oh(1)}$-time algorithm for some computable function~$g$ (unless~$\text{FPT}=\text{W[1]}$) and a stronger ETH-based running~time bound.

\begin{theorem} \label{Theorem: Hardness for h}
\textsc{Polytree Learning} is W[1]-hard when parameterized by the number of dependent vertices~$d$; if the ETH holds, then it has no~${(\delta_\Fa)}^{o(d)} \cdot |I|^{\Oh(1)}$-time algorithm. Both results even hold for instances where the directed superstructure~$S_\Fa$ is a DAG.
\end{theorem}

\newcommand{\proofHardnessForH}{
Next, we show that~$I$ is a yes-instance of~\IS if and only if~$I'$ is a yes-instance of~\PL.

$(\Rightarrow)$
Let~$S\subseteq V$ be an independent set of size~$k$ in~$G$.
Moreover, let~$S = \{u_1, \dots, u_k\}$.
We set~$A := \cup_{i= 1}^k P_{u_i} \times v_i$ and show that~$D=(N,A)$ is a polytree of score~$t=k$.
Since~$S$ is an independent set in~$G$, for every edge~$e\in E$,~$w_e$ is the endpoint of at most one arc in~$A$.
Moreover, there is no arc~$(v_i, v_j)$ contained in~$A$ and~$v^*$ is only an endpoint of the arcs~$(v^*, v_i)$ for each~$i\in[1,k]$.
As a consequence,~$D$ is a polytree and has score~$k$ due to the fact that~$f_{v_i}(P_{u_i}) = 1$ for all~$i\in [1,k]$.
Hence,~$I'$ is a yes-instance of~\PL.

$(\Leftarrow)$
Let~$A\subseteq N \times N$ be an arc set such that~$D=(N,A)$ is a polytree with score at least~$t=k$.
By construction, only the vertices~$v_1, \dots, v_k$ can learn a nonempty potential parent set.
Moreover, no parent set has score larger than one.
As a consequence, for every~$i\in [1,k]$ there is some~$u_i\in V$ such that~$v_i$ learns the parent set~$P_{u_i}$ under~$A$.
We show that~$S:=\{u_1, \dots, u_k\}$ is an independent set of size~$k$ in~$G$.
By construction of the local scores, each parent set~$P_i$ contains the vertex~$v^*$.
Since every vertex of~$G$ has degree at least one, each such parent set has size at least two.
Hence, the vertices~$u_i$ and~$u_j$ are distinct if~$i \neq j$ as, otherwise, the skeleton of~$D$ would contain the cycle~$(v_i, v^*, v_j, w_e)$ for each~$w_e  \in P_{u_i} \setminus \{v^*\}$.
Moreover, since~$D$ is a polytree, distinct vertices~$u_i$ and~$u_j$ are not adjacent in~$G$ as, otherwise, the vertex~$w_{\{u_i, u_j\}}$ is contained in both the learned parent sets~$P_{u_i}$ and~$P_{u_j}$ and, hence, the cycle~$(v_i, v^*, v_j, w_{\{u_i, u_j\}})$ is contained in the skeleton of~$D$.
Consequently,~$I$ is a yes-instance of~\IS since~$S$ is an independent set of size~$k$ in~$G$.
}

\begin{proof}
We reduce from~\IS where one is given an undirected graph~$G=(V,E)$ and an integer~$k$ and the question is whether there is a subset~$S\subseteq V$ of size at least~$k$ such that no two vertices in~$S$ are connected by an edge. \IS is~\W1-hard when parameterized by~$k$~\iflong\cite{DF13}\else\cite{CFKLMPPS15}\fi.

Given an instance~$I=(G=(V,E),k)$ of~\IS, we describe how to construct an equivalent instance~$I'=(N,\Fa, t)$ of~\PL in polynomial time such that at most~$k$ vertices have a nonempty potential parent set.
Note that we can assume that every vertex of~$G$ has degree at least one.
We start with an empty set~$N$ and add~$k+1$ vertices~$v_1, \dots, v_k,$ and~$v^*$. 
Moreover, we add a vertex~$w_e$ for each edge~$e\in E$.
For every vertex~$v\in V$, we set~$P_v := \{w_{\{v,u\}} \mid u\in N_G(v)\} \cup \{v^*\}$ and we set~$f_{v_i}(P_v) := 1$ for each~$i\in[1,k]$.
All other local scores are set to 0.
Finally, we set~$t:=k$. 
This completes the construction of~$I'$.
\iflong 
\proofHardnessForH
\else
We omit the correctness proof.
\fi

In the constructed instance,~$d = k$ and~$\delta_\Fa = n +1$.
Unless the ETH fails,~\IS cannot be solved in~$n^{o(k)}$ time~\cite{CHKX06} and, hence, \PL cannot be solved in~$(\delta_\Fa)^{o(d)}\cdot |I|^{\Oh(1)}$ time. 
\end{proof}

Theorem~\ref{Theorem: Hardness for h} points out a difference between~\PL and~\textsc{Bayesian Network Structure Learning (BNSL)}, where we aim to learn a DAG. In~BNSL, a nondependent vertex~$v$ can be easily removed from the input instance~$(N,\Fa,t)$ by setting~$N':=N \setminus \{v\}$ and modifying the local scores to~$f'_u(P):= \max (f_u(P), f_u(P \cup \{v\}))$.  

\section{Dependent Vertices and Small Parent Sets}\label{sec:super matroid}
Due to Theorem~\ref{Theorem: Hardness for h}, fixed-parameter tractability for \textsc{Polytree Learning} parameterized by~$d$ is presumably not possible. However, in instances constructed in the proof of Theorem~\ref{Theorem: Hardness for h} the maximum parent set size~$p$ is not bounded by some computable function in~$d$. In practice there are many instances where~$p$ is relatively small or upper-bounded by some small constant~\cite{BH15}. First, we provide an FPT algorithm for the parameter~$d+p$\iflong , the sum of the number of dependent vertices and the maximum parent set size\fi. Second, we provide a polynomial kernel for the parameter~$d$ if the maximum parent set size~$p$ is constant. Both results are based on computing max~$q$-representative sets in a matroid~\cite{FLS14,LMPS18}.

To apply the technique of representative sets we assume that there is a solution with exactly~$d \cdot p$ arcs and every nonempty potential parent set contains exactly~$p$ vertices. This can be obtained with the following simple modification of an input instance~$(N,\Fa,t)$: For every dependent vertex~$v$ we add vertices~$v_1, v_2, \dots, v_p$ to~$N$ and set~$f_{v_i}(P):=0$ for all their local scores. Then, for every potential parent set~$P \in P_\Fa(v)$ with~$|P|<p$ we set~$f_v(P \cup \{v_1, \dots, v_{p-|P|}\}) := f_v(P)$ and, afterwards, we set~$f_v(P):=0$. Then, the given instance is a yes-instance if and only if the modified instance has a solution with exactly~$d\cdot p$ arcs. Furthermore, note that~$f_v(\emptyset)=0$ for every dependent vertex and every nonempty potential parent set has size exactly~$p$ after applying the modification. Before we present the results of this section we state the definition of a~matroid.

\begin{definition} \label{Definition: Matroid}
A pair~$M=(E,\mathcal{I})$, where~$E$ is a set and~$\mathcal{I}$ is a family of subsets of~$E$ is a \emph{matroid} if
\begin{enumerate}[1.]
\item $\emptyset \in \mathcal{I}$,
\item if $A \in \mathcal{I}$ and~$B \subseteq A$, then~$B \in \mathcal{I}$, and
\item if~$A,B \in \mathcal{I}$ and~$|A|<|B|$, then there exists some~$b \in B \setminus A$ such that~$A \cup \{b\} \in \mathcal{I}$.
\end{enumerate}
\end{definition}

Given a matroid~$M=(E,\Ind)$, the sets in~$\Ind$ are called \emph{independent sets}. 
A \emph{representation of~$M$ over a field~$\mathds{F}$} is a mapping~$\varphi:E \rightarrow V$ where~$V$ is some vector space over~$\mathds{F}$ such that~$A \in \Ind$ if and only if the vectors~$\varphi(a)$ with~$a \in A$ are linearly independent in~$V$. A matroid with a representation is called~\emph{linear matroid}. Given a set~$B \subseteq E$, a set~$A \subseteq E$ \emph{fits}~$B$ if~$A \cap B = \emptyset$ and~$A \cup B \in \mathcal{I}$.

\begin{definition}
Let~$M=(E,\Ind)$ be a matroid, let~$\mathcal{A}$ be a family of subsets of~$E$, and let~$w: \mathcal{A} \rightarrow \mathds{N}_0$ be a weight function. A subfamily~$\widehat{\mathcal{A}} \subseteq \mathcal{A}$ \emph{max $q$-represents}~$\mathcal{A}$ (with respect to~$w$) if for every set~$B \subseteq E$ with~$|B|=q$ the following holds: If there is a set~$A \in \mathcal{A}$ that fits~$B$, there exists some~$\widehat{A} \in \widehat{\mathcal{A}}$ that fits~$B$, and~$w(\widehat{A}) \geq w(A)$. If~$\widehat{\mathcal{A}}$ max~$q$-represents~$\mathcal{A}$ we write~$\widehat{\mathcal{A}} \subseteq^q \mathcal{A}$.
\end{definition}

We refer to a set family~$\mathcal{A}$ where every~$A \in \mathcal{A}$ has size exactly~$x \in \mathds{N}_0$ as an~\emph{$x$-family}. Our results rely on the fact that max~$q$-representative sets of an~$x$-family can be computed efficiently as stated in a theorem by Lokshtanov et al.~\shortcite{LMPS18} that is based on an algorithm by Fomin et al.~\shortcite{FLS14}. In the following,~$\omega < 2.373$ is the matrix multiplication constant~\cite{W12}.

\begin{theorem}[\cite{LMPS18}] \label{Theorem: Repr Sets Algo}
Let~$M=(E,\Ind)$ be a linear matroid which representation can be encoded with a~$k \times |E|$ matrix over the field~$\mathds{F}_2$ for some~$k \in \mathds{N}$. Let~$\mathcal{A}$ be an~$x$-family containing~$\ell$ sets, and let~$w: \mathcal{A} \rightarrow \mathds{N}_0$ be a weight function. Then,
\begin{enumerate}
\item[a)] there exists some~$\widehat{\mathcal{A}} \subseteq^q \mathcal{A}$ of size~$\binom{x+q}{x}$ that can be computed with~$\Oh\left( \binom{x+q}{x}^2 \cdot \ell x^3 k^2 + \ell \binom{x+q}{q}^{\omega}kx \right)+(k+|E|)^{\mathcal{O}(1)}$ operations in~$\mathds{F}_2$, and
\item[b)] there exists some~${\widehat{\mathcal{A}} \subseteq^q \mathcal{A}}$~of~size~${\binom{x+q}{x} \cdot k \cdot x}$ that~can be~computed~with~$\Oh \Big( \binom{x+q}{x} \cdot \ell x^3 k^2 + \ell \binom{x+q}{q}^{\omega-1}(kx)^{\omega-1}\Big )+(k+|E|)^{\mathcal{O}(1)}$~operations in~$\mathds{F}_2$.
\end{enumerate} 
\end{theorem}

We next define the matroid we use in this work. Recall that, given an instance~$(N,\Fa,t)$ of \textsc{Polytree Learning}, the directed superstructure~$S_\Fa$ is defined as~$S_\Fa:=(N,A_\Fa)$  where~$A_\Fa$ is the set of arcs that are potentially present in a solution, and we set~$m:=|A_\Fa|$. In this work we consider the \emph{super matroid} $M_\Fa$ which we define as the graphic matroid~\cite{T65} of the super structure. Formally, $M_\Fa:=(A_\Fa, \Ind)$ where~$A \subseteq A_\Fa$ is independent if and only~$(N,A)$ is a polytree. \iflong

\fi The super matroid is closely related to the~\emph{acyclicity matroid} that has been used for a constrained version of~\textsc{Polytree Learning}~\cite{GKLOS15}. The proof of the following proposition is along the lines of the proof that the graphic matroid is a linear matroid. \iflong We provide it here for sake of completeness. \fi

\begin{proposition}\label{prop: super linear matroid}
Let~$(N,\Fa,t)$ be an instance of \textsc{Polytree Learning}. Then, the super matroid~$M_\Fa$ is a linear matroid and its representation can be encoded by an~$n \times m$ matrix over the field~$\mathds{F}_2$.
\end{proposition}

\newcommand{\proofSuperLinearMatroid}{
\begin{proof}
We first show that~$M_\Fa$ is a matroid. Since~$(N,\emptyset)$ contains no cycles, it holds that~$\emptyset \in \Ind$. Next, if~$(N,A)$ is a polytree for some~$A \subseteq A_\Fa$, then~$(N,B)$ is a polytree for every~$B \subseteq A$. Thus, Conditions~$1$ and~$2$ from Definition~\ref{Definition: Matroid}~hold. 

We next show that Condition~$3$ holds. Consider~$A,B \in \Ind$ with~$|A| < |B|$. Let~$N' \subseteq N$ be the vertices of a connected component of~$(N,A)$. Since~$(N,B)$ is a polytree, the number of arcs in~$B$ between vertices of~$N'$ is at most the number of arcs in~$A$ between the vertices of~$N'$. Then, since~$|B|>|A|$, there exists some~$(u,v) \in B \setminus A$ that has endpoints in two distinct connected components of~$(N,A)$. Thus, $(N,A \cup \{(u,v)\})$ is a polytree. Consequently,~$M_\Fa$ is a matroid.

We next consider the representation of~$M_\Fa$. Let~$v_1, \dots, v_n$ be the elements of~$N$. We define the mapping~$\varphi :A_\Fa \rightarrow {\mathds{F}_2}^n$ by letting~$\varphi((v_i,v_j))$ be the vector where the~$i$-th and the~$j$-th entry equal~$1$ and all other entries equal~$0$. It is easy to see that an arc-set~$A$ is independent in~$M_\Fa$ if and only if the~$\varphi(a)$ with~$a \in A$ are linearly independent in~${\mathds{F}_2}^n$. Clearly,~$\varphi$ can be encoded by an~$n \times m$ matrix over~$\mathds{F}_2$.
\end{proof}
}
\iflong
\proofSuperLinearMatroid
\fi

\iflong
\subsection{An FPT Algorithm}
\else
\paragraph{An FPT Algorithm.}
\fi
We now use the super matroid~$M_\Fa$ to show that \textsc{Polytree Learning} can be solved in~$2^{\omega dp} \cdot |I|^{\Oh(1)}$~time where~$\omega$ is the matrix multiplication constant. The idea of the algorithm is simple: Let~$H:=\{v_1, \dots, v_d\}$ be the set of dependent vertices, and for~$i \in \{0,1, \dots, d\}$ let~$H_i:=\{v_1, \dots, v_i\}$ be the set containing the first~$i$ dependent vertices. The idea is that, for every~$H_i$, we compute a family~$\mathcal{A}_i$ of possible polytrees where only the vertices from~$\{v_1, \dots, v_i\}$ learn a nonempty potential parent set. We use the algorithm behind Theorem~\ref{Theorem: Repr Sets Algo} as a subroutine to delete arc-sets from~$\mathcal{A}_i$ that are not necessary to find a solution. We next define the operation~$\oplus$. Intuitively, $\mathcal{A} \oplus_v P$ means that we extend each possible solution in the family~$\mathcal{A}$ by the arc-set that defines~$P$ as the parent set of a vertex~$v$.

\begin{definition}
Let~$v \in N$, and let~$\mathcal{A}$ be an~$x$-family of subsets of~$A_\Fa$ such that~$P^A_v = \emptyset$ for every~$A \in \mathcal{A}$. For a vertex set~$P \subseteq N$ we define
\begin{align*}
\mathcal{A} \oplus_v P := \{A \cup (P \times v) \mid A \in \mathcal{A} \text{ and } A \cup (P \times v) \in \Ind\}.
\end{align*}
\end{definition}

Observe that for every~$A \in \mathcal{A}$, the set~$P \times v$ is disjoint from~$A$ since~$P^A_v = \emptyset$. Consequently,~$\mathcal{A} \oplus_v P$ is an~$(x+|P|)$-family. The next lemma ensures that some operations (including~$\oplus$) are compatible with representative sets.

\begin{lemma} \label{Lemma: compatible with repr sets}
Let~$w: 2^{A_\Fa} \rightarrow \mathds{N}_0$ be a weight function with~$w(A):=\score(A)$. Let~$\mathcal{A}$ be an~$x$-family of subsets of~$A_\Fa$.
\begin{enumerate}[a)]
\item If~$\widetilde{\mathcal{A}} \subseteq^q \mathcal{A}$ and~$\widehat{\mathcal{A}} \subseteq^q \widetilde{\mathcal{A}}$, then~$\widehat{\mathcal{A}}\subseteq^q \mathcal{A}$.
\item If~$\widehat{\mathcal{A}} \subseteq^q \mathcal{A}$ and~$\mathcal{B}$ is an~$x$-family of subsets of~$A_\Fa$ with~$\widehat{\mathcal{B}} \subseteq^q \mathcal{B}$, then~$\widehat{\mathcal{A}} \cup \widehat{\mathcal{B}} \subseteq^q \mathcal{A} \cup \mathcal{B}$.
\item Let~$v \in N$ and let~$P \subseteq N$ such that~$P^A_v=\emptyset$ for every~$A \in \mathcal{A}$. Then, if~$\widehat{\mathcal{A}} \subseteq^{q+|P|} \mathcal{A}$ it follows that~$\widehat{\mathcal{A}} \oplus_v P \subseteq^{q} \mathcal{A} \oplus_v P$.
\end{enumerate}
\end{lemma}
\newcommand{\proofCompReprSets}{
\begin{proof}
Statements~$a)$ and~$b)$ are well-known facts~\cite{CFKLMPPS15,FLS14}. We prove Statement~$c)$. Let~$B$ be a set of size~$q$. Let there be a set~$A \cup (P \times v) \in \mathcal{A} \oplus_v P$ that fits~$B$. That is,
\begin{align}
&B \cap (A \cup (P \times v)) =\emptyset, \text{ and} \label{Property 1}\\
&B \cup (A \cup (P \times v)) \in \Ind. \label{Property 2}
\end{align}
We show that there is some~$\widehat{A} \cup (P \times v) \in \widehat{\mathcal{A}} \oplus_v P$ that fits~$B$ and~$w(\widehat{A} \cup (P \times v)) \geq w(A \cup (P \times v))$. To this end, observe that Property~(\ref{Property 1}) implies
\begin{align}
&B \cap A = \emptyset, \text{ and } \label{Property *}\\
&B \cap (P \times v) = \emptyset. \label{Property **}
\end{align}
We define~$\overline{B}:=B \cup (P \times v)$. Observe that Property~(\ref{Property *}) together with~$A \cap (P \times v)=\emptyset$ implies~$\overline{B} \cap A = \emptyset$, and that Property~(\ref{Property 2}) implies~$A \cup \overline{B} \in \Ind$. Consequently,~$A$ fits~$\overline{B}$. Then, since~$|\overline{B}|=q+|P|$ due to Property~(\ref{Property **}) and~$\widehat{\mathcal{A}} \subseteq^{q+|P|} \mathcal{A}$, there exists some~$\widehat{A} \in \widehat{\mathcal{A}}$ that fits~$\overline{B}$ and~$w(\widehat{A}) \geq w(A)$.

Consider~$\widehat{A} \cup (P \times v)$. Since~$\widehat{A}$ fits~$\overline{B}$ it holds that~$(\widehat{A} \cup (P \times v)) \cup B \in \Ind$. Moreover, observe that~$w(\widehat{A} \cup (P \times v)) = w( \widehat{A}) + f_v(P) \geq w(A) + f_v(P) = w(A \cup (P \times v))$. It remains to show that~$(\widehat{A} \cup (P \times v)) \cap B = \emptyset$. Since~$\widehat{A}$ fits~$\overline{B}$ and~$B \subseteq \overline{B}$ it holds that~$\widehat{A} \cap B = \emptyset$. Then,~Property~(\ref{Property **}) implies~$(\widehat{A} \cup (P \times v)) \cap B = (\widehat{A} \cap B) \cup ( (P \times v) \cap B)) = \emptyset$. Thus,~$\widehat{A} \cup (P \times v)$ fits~$B$ and therefore~$\widehat{\mathcal{A}} \oplus_v P \subseteq^{q} \mathcal{A} \oplus_v P$.
\end{proof}
}
\iflong
\proofCompReprSets
\fi
We now describe the FPT algorithm. Let~$I:=(N,\Fa,t)$ be an instance of \textsc{Polytree Learning}. Let~$H:=\{v_1, v_2, \dots, v_d\}$ denote the set of dependent vertices of~$I$, and for~$i \in \{0,1, \dots, d\}$ let~$H_i:=\{v_1, \dots, v_i\}$ contain the first~$i$ dependent vertices. Observe that~$H_0=\emptyset$ and~$H_d=H$. We define~$\mathcal{A}_i$ as the family of possible directed graphs (even the graphs that are no polytrees) where only the vertices in~$H_i$ learn a nonempty potential parent set. Formally, this is 
\begin{align*}
\mathcal{A}_i := \left\lbrace A \subseteq A_\Fa \middle \vert \begin{array}{l}
    P^A_{v} \in \mathcal{P}_\Fa(v) \setminus \{\emptyset\} \text{ for all }v \in  H_i\\
    P^A_{v}=\emptyset \text{ for all }v\in N \setminus H_i
  \end{array}
 \right\rbrace .
\end{align*}

The  algorithm is based on the following recurrence.

\begin{lemma} \label{Lemma: Recurrence for Setfamily}
If~$i=0$, then~$\mathcal{A}_i = \{ \emptyset \}$. If~$i>0$,~$\mathcal{A}_i$ can be computed by
\iflong
\begin{align*}
\mathcal{A}_i = \bigcup_{P \in \mathcal{P}_\Fa(v_i) \setminus \{\emptyset\}} \mathcal{A}_{i-1} \oplus_{v_i} P.
\end{align*}
\else
$\mathcal{A}_i = \bigcup_{P \in \mathcal{P}_\Fa(v_i) \setminus \{\emptyset\}} \mathcal{A}_{i-1} \oplus_{v_i} P$.
\fi
\end{lemma}

\begin{algorithm}[t]
\KwIn{$(N,\Fa,t)$ and dependent vertices~$v_1, \dots, v_d$}
$\widehat{\mathcal{A}}_0 := \{\emptyset\}$ \label{Line: Before Loop}\\
\For{$i=1 \dots d$}{
	$\widetilde{\mathcal{A}}_i = \bigcup_{P \in \mathcal{P}_\Fa(v_i) \setminus \{\emptyset\}} \widehat{\mathcal{A}}_{i-1} \oplus_{v_i} P$ \label{Line: Recurrence}\\
	$\widehat{\mathcal{A}}_i:=$ \texttt{ComputeRepresentation}$(\widetilde{\mathcal{A}}_i, (d-i) \cdot p)$ \label{Line: Subroutine}
}
\textbf{return} $\widehat{A} \in \widehat{\mathcal{A}}_d$~such that~$(N,\widehat{A})$ is a polytree and~$\score(\widehat{A})$ is maximal\label{Line: Return Statement}

\caption{FPT-algorithm for parameter~$d+p$} \label{Algorithm: FPT h+p}
\end{algorithm}

Intuitively, Lemma~\ref{Lemma: Recurrence for Setfamily} states that~$\mathcal{A}_i$ can be computed by considering~$\mathcal{A}_{i-1}$ and combining every~$A \in \mathcal{A}_{i-1}$ with every arc-set that defines a nonempty potential parent set of~$v_i$. The correctness proof is straightforward and thus omitted.
%
%
%
We next present the FPT algorithm.

\begin{theorem}
\textsc{Polytree Learning} can be solved in~$2^{\omega d p} \cdot |I|^{\Oh(1)}$~time\iflong, where~$\omega$ is the matrix multiplication constant. \else. \fi
\end{theorem}

\newcommand{\proofClaimOne}{

The loop-invariant holds before entering the loop since~$\widehat{\mathcal{A}}_0:= \{\emptyset\}$ in Line~\ref{Line: Before Loop} and~$\mathcal{A}_0=\{\emptyset\}$. Suppose that the loop-invariant hold for the~$(i-1)$th execution of the loop. We show that the loop-invariant holds after the~$i$th execution.

First, consider Line~\ref{Line: Recurrence}. Since we assume that every nonempty potential parent set contains exactly~$p$ vertices, Lemma~\ref{Lemma: compatible with repr sets} and Lemma~\ref{Lemma: Recurrence for Setfamily} imply that~$\widetilde{\mathcal{A}}_i$ max~$((d-i)\cdot p)$-represents~$\mathcal{A}_i$. Note that at this point~$\widetilde{A}_i$ contains up to~$\max(1,\binom{dp}{(i-1) p} \cdot n \cdot (i-1) \cdot p) \cdot \delta_\Fa$~sets.

Next, consider Line~\ref{Line: Subroutine}. Since we assume that every nonempty potential parent set contains exactly~$p$ vertices, the family~$\widetilde{A}_i$ is an~$(i \cdot p)$-family. Then, the algorithm behind Theorem~\ref{Theorem: Repr Sets Algo} computes a~$((d-i)\cdot p)$-representing family~$\widehat{\mathcal{A}}_i$. Then, by Theorem~\ref{Theorem: Repr Sets Algo}~$b)$ and Lemma~\ref{Lemma: compatible with repr sets}~$a)$, $\widehat{\mathcal{A}}_i$ max~$((d-i)\cdot p)$-represents~$\mathcal{A}_i$ and~$|\widehat{\mathcal{A}}_i| \leq \binom{dp}{ip} \cdot n\cdot i \cdot p$ after the execution of Line~\ref{Line: Subroutine}.}

\begin{proof}
Let~$I:=(N,\Fa,t)$ be an instance of \textsc{Polytree Learning} with dependent vertices~$H=\{v_1, \dots, v_d\}$, let the families~$\mathcal{A}_i$ for~$i \in \{0,1, \dots,d\}$ be defined as above, and let~$w: 2^{A_\Fa} \rightarrow \mathds{N}_0$ be defined by~$w(A):= \score(A)$. All representing families considered in this proof are max~representing families with respect to~$w$. We prove that Algorithm~\ref{Algorithm: FPT h+p} computes an arc-set~$A$ such that~$(N,A)$ is a polytree with maximal score. 

The subroutine~\texttt{ComputeRepresentation}$(\widetilde{\mathcal{A}}_i, (d-i) \cdot p)$ in Algorithm~\ref{Algorithm: FPT h+p} is an application of the algorithm behind Theorem~\ref{Theorem: Repr Sets Algo}~$b)$. It computes a max~$((d-i) \cdot p)$-representing family for~$\widetilde{\mathcal{A}}_i$. As a technical remark we mention that the algorithm as described by Lokshtanov et al.~\shortcite{LMPS18} evaluates the weight~$w(A)$ for~$|\widetilde{\mathcal{A}}_i|$ many arc-sets~$A$. We assume that each such evaluation~$w(A)$ is replaced by the computation of~$\score(A) = \sum_{v \in H} f_v(P^A_v)$ which can be done in~$|I|^{\Oh(1)}$~time.

\textit{Correctness.} We first prove  the following invariant.
\begin{claim} \label{Claim: Loop Invariant}
 The family~$\widehat{\mathcal{A}}_i$ max~$((d-i)\cdot p)$-represents~$\mathcal{A}_i$ and~$|\widehat{\mathcal{A}}_i| \leq \max (1,\binom{d p}{ip} \cdot n\cdot i \cdot p)$.
\end{claim}
\begin{claimproof} \proofClaimOne $\hfill \Diamond$
\end{claimproof}

We next show that~$\widehat{\mathcal{A}}_d$ contains an arc-set that defines a polytree with maximum score and thus, a solution is returned in Line~\ref{Line: Return Statement}. Since we assume that there is an optimal solution~$A$ that consists of exactly~$d \cdot p$ arcs, this solution is an element of the family~$\mathcal{A}_d$. Then, since~$\widehat{\mathcal{A}}_d \subseteq^0 \mathcal{A}_d$, there exists some~$\widehat{A} \in \widehat{\mathcal{A}}_d$ with~$\widehat{A} \cup \emptyset \in \mathcal{I}$ and~$w(\widehat{A}) \geq w(A)$. Since~$\widehat{A} \cup \emptyset \in \mathcal{I}$, the graph~$(N,\widehat{A})$ is a polytree, and since~$w(\widehat{A}) \geq w(A)$ the score of~$\widehat{A}$ is maximal.

\textit{Running time.} We next analyze the running time of the algorithm. For this analysis, we use the inequality~$\binom{a}{b} \leq 2^a$ for every~$b\leq a$. Let~$i$ be fixed.

We first analyze the running time of one execution of Line~\ref{Line: Recurrence}. Since~$\widehat{\mathcal{A}}_{i-1}$ has size at most~$\binom{dp}{(i-1)p} \cdot n \cdot i \cdot p$ due to Claim~\ref{Claim: Loop Invariant}, Line~\ref{Line: Recurrence} can be executed in~$2^{dp} \cdot |I|^{\Oh(1)}$~time. 

We next analyze the running time of one execution of Line~\ref{Line: Subroutine}. Recall that~$\widetilde{\mathcal{A}}_{i}$ is an~$(i \cdot p)$-family of size at most~$\binom{dp}{(i-1)p} \cdot n \cdot i \cdot p \cdot \delta_\Fa$. Furthermore, recall that there are~$|\widetilde{\mathcal{A}}_i|$~many evaluations of the weight function. Combining the running time from Theorem~\ref{Theorem: Repr Sets Algo}~$b)$ with the time for evaluating~$w$, the subroutine takes time
\begin{align*}
\begin{split}
&\Oh \left( \binom{dp}{ip} \binom{dp}{(i-1)p} \delta_\Fa (i \cdot p)^4 n^3 \right.\\
&~~~~~+  \left. \binom{dp}{(i-1)p} \delta_\Fa \binom{dp}{ip}^{\omega-1} (n \cdot i \cdot p)^{\omega} \right)
 \end{split}\\
& + (n + m)^{\Oh(1)}\\
& +  \underbrace{\binom{dp}{(i-1)p} \cdot n \cdot i \cdot p \cdot  \delta_\Fa \cdot |I|^{\Oh(1)}}_{\text{evaluating }w}.
\end{align*}
Therefore, one execution of Line~\ref{Line: Subroutine} can be done in~$2^{\omega dp} |I|^{\Oh(1)}$~time. Since there are~$d$~repetitions of Lines~\ref{Line: Recurrence}--\ref{Line: Subroutine}, and Line~\ref{Line: Return Statement} can be executed in~$|I|^{\Oh(1)}$ time, the algorithm runs within the claimed running~time.\end{proof}

\iflong
\subsection{Problem Kernelization}
\else
\paragraph{Problem Kernelization.}
\fi

We now study problem kernelization for \textsc{Polytree Learning} parameterized by~$d$
 when the maximum parent set size~$p$ is constant. We provide a problem kernel consisting of at most~$(dp)^{p+1}+d$ vertices where each vertex has at most~$(dp)^p$ potential parent sets which can be computed in~$(dp)^{{\omega p}} \cdot |I|^{\Oh(1)}$ time. Observe that both, the running time and the kernel size, are polynomial for every constant~$p$.
Note also that, since 
~$d+p \in \Oh(n)$, Theorem~\ref{Theorem: No-Poly for n} implies that there is presumably no kernel of size~$(d+p)^{\Oh(1)}$ that can be computed in~$(d+p)^{\Oh(1)}$~time.

 The basic idea of the kernelization  is that we use max~$q$-representations to identify nondependent vertices that are not necessary to find a solution.

\begin{theorem} \label{Theorem: 'Kernel'}
There is an algorithm  that, given an instance~$(N,\Fa,t)$ of \textsc{Polytree Learning} computes in time~$(dp)^{\omega p} \cdot |I|^{\Oh(1)}$ an equivalent instance~$(N',\Fa',t)$ such that~$|N'| \leq (dp)^{p+1}+d$ and~$\delta_{\Fa'} \leq (dp)^p$. 
\end{theorem}

\begin{proof}
Let~$H$ be the set of dependent vertices\iflong of~$(N,\Fa,t)$\fi.

\textit{Computation of the reduced instance.} We describe how we compute~$(N',\Fa',t)$. We define the family~$\mathcal{A}_v := \{P \times v \mid P \in \mathcal{P}_\Fa(v) \}$ for every~$v \in H$ and the weight function~$w:\mathcal{A}_v \rightarrow \mathds{N}_0$ by~$w(P \times v):= f_v(P)$. We then apply the algorithm behind Theorem~\ref{Theorem: Repr Sets Algo}~$a)$ and compute a max~$((d-1)\cdot p)$-representing family~$\widehat{\mathcal{A}}_v$ for every~$\mathcal{A}_v$.

Given all~$\widehat{\mathcal{A}}_v$, a vertex~$w$ is called \emph{necessary} if~$w \in H$ or if there exists some~$v \in H$ such that~$(w,v) \in A$ for some~$A \in \widehat{\mathcal{A}}_v$. We then define~$N'$ as the set of necessary vertices. Next,~$\Fa'$ consists of local score functions~$f_v':2^{N'\setminus \{v\}} \rightarrow \mathds{N}_0$ with~$f_v'(P) := f_v(P)$ for every~$P \in 2^{N' \setminus \{v\}}$. In other words,~$f_v'$ is the limitation of~$f_v$ on parent sets that contain only necessary vertices.

Next, consider the running-time of the computation of~$(N',\Fa',t)$. Since each~$\mathcal{A}_v$ contains at most~$\delta_\Fa$ arc sets and we assume that every potential parent set has size exactly~$p$, each~$\widehat{\mathcal{A}}_v$ can be computed in time
\begin{equation*}
\resizebox{.91\linewidth}{!}{$
    \displaystyle
    \mathcal{O}\left( \binom{dp}{p}^2 \cdot \delta_\Fa \cdot p^3 \cdot n^2 + \delta_{\Fa} \cdot \binom{dp}{p}^{\omega} \cdot n \cdot p \right) + (n+m)^{\Oh(1)}.
$}
\end{equation*}
Observe that we use the symmetry of the binomial coefficient to analyze this running~time. After computing all~$\widehat{\mathcal{A}}_v$, we compute~$N'$ and~$\Fa'$ in polynomial time in~$|I|$. The overall running time is~$(dp)^{\omega p} \cdot |I|^{\Oh(1)}$. 

\textit{Correctness.} We next show that~$(N,\Fa,t)$ is a yes-instance if and only if~$(N',\Fa',t)$ is a yes-instance.

$(\Leftarrow)$ Let~$(N',\Fa',t)$ be a yes-instance. Then, there exists an arc-set~$A'$ such that~$(N',A')$ is a polytree with score at least~$t$. Since~$N' \subseteq N$,~$f_v'(P)=f_v(P)$ for every~$v \in N'$, and~$P \subseteq N' \setminus \{v\}$ we conclude that~$(N,A')$ is a polytree with~score at least~$t$.

$(\Rightarrow)$ Let~$(N,\Fa,t)$ be a yes-instance. We choose a solution~$A$~for~$(N,\Fa,t)$ such that~$P^A_v \subseteq N'$ for as many dependent vertices~$v$ as possible. We prove that this implies that~$P^A_v \subseteq N'$ for all dependent vertices. Assume towards a contradiction that there is some~$v \in H$ with~$P^A_v \not \subseteq N'$. Observe that~$(P^A_v \times v) \in \mathcal{A}_v \setminus \widehat{\mathcal{A}}_v$. We then define the arc set~$B:= \bigcup_{w \in H \setminus \{v\}} P^A_w \times w$. Since we assume that all nonempty potential parent sets have size exactly~$p$, we conclude~$|B|=(d-1)p$. Then, since~$\widehat{\mathcal{A}}_v$ max~$((d-1)\cdot p)$-represents~$\mathcal{A}_v$ and~$(P^A_v \times v) \in \mathcal{A}_v$ fits~$B$ we conclude that there is some~$(P \times v) \in \widehat{\mathcal{A}}_v$ such that~$B \cap (P \times v) = \emptyset$, $(N,B \cap (P \times v))$ is a polytree, and~$f_v(P) \geq f_v(P^A_v)$. Thus,~$C:=B\cup (P \times v)$ is a solution of~$(N,\Fa,t)$ and the number of vertices~$v$ that satisfy~$P^C_v \subseteq N'$ is bigger than the number of vertices~$v$ that satisfy~$P^A_v \subseteq N'$. This contradicts the choice of~$A$.

\textit{Bound on the size of~$|N'|$ and~$\delta_{\Fa'}$.} By Theorem~\ref{Theorem: Repr Sets Algo}, each~$\widehat{\mathcal{A}}_i$ has size at most~$\binom{(d-1) p + p}{p}=\binom{dp}{p}$. Consequently,~$\delta_{\Fa'} \leq (dp)^p$ and~$N'\leq d \cdot \binom{dp}{p} \cdot p +d \leq {(dp)^{p+1}+d}$.
\end{proof}

\newcommand{\proofAndContent}{
Observe that the instance~$I:=(N',\Fa',t)$ from Theorem~\ref{Theorem: 'Kernel'} is technically not a kernel since the encoding of the integer~$t$ and the values of~$f_v(P)$ might not be bounded in~$d$ and thus the size of the instance~$|I|$ is not bounded in~$d$. We use the following lemma~\iflong\cite{EKMR17,FT87}\else\cite{EKMR17}\fi to show that Theorem~\ref{Theorem: 'Kernel'} implies an actual polynomial kernel for the parameter~$d$ when~$p$ is constant.

\begin{lemma}[\cite{EKMR17}] \label{Lemma: shrink score}
There is an algorithm that, given a vector~$w \in \mathds{Q}^r$ and some~$W \in \mathds{Q}$ computes in polynomial time a vector~$\overline{w}=(w_1, \dots, w_r) \in \mathds{Z}^r$ where~$\max_{i\in \{1, \dots, r\}} |w_i| \in 2^{\Oh(r^3)}$ and an integer~$\overline{W} \in \mathds{Z}$ with total encoding length~$\Oh(r^4)$ such that~$w \cdot x \geq W$ if and only if~$\overline{w} \cdot x \geq \overline{W}$ for every~$x \in \{0,1\}^r$.
\end{lemma}

\begin{corollary}
\textsc{Polytree Learning} with constant parent set size admits a polynomial kernel when parameterized by~$d$.
\end{corollary}
\begin{proof}
Let~$(N',\Fa',t)$ be the reduced instance from Theorem~\ref{Theorem: 'Kernel'}. Let~$r$ be the number of triples~$(P,|P|, f_v(P))$ in the two-dimensional array representing~$\Fa'$. Clearly,~$r \leq |N'|\cdot \delta_{\Fa'}$. Let~$w$ be the~$r$-dimensional vector containing all values~$f_v(P)$ for all~$v$ and~$P$. Applying the algorithm behind Lemma~\ref{Lemma: shrink score} on~$w$ and~$t$ computes a vector~$\overline{w}$ and an integer~$\overline{t}$ that has encoding length~$\Oh((|N'| \cdot \delta_{\Fa'})^4)$ with the property stated in~Lemma~\ref{Lemma: shrink score}. 

Substituting all local scores stored in~$\Fa'$ with the corresponding values in~$\overline{w}$ and substituting~$t$ by~$\overline{t}$ converts the instance~$(N',\Fa',t)$ into an equivalent instance which size is polynomially bounded in~$d$ if~$p$ is constant.
\end{proof}
}
\iflong
\proofAndContent
\fi

\section{Conclusion}
We believe that there is potential for practically relevant exact \textsc{Polytree
  Learning} algorithms and that this work could constitute a first step. Next,
one should aim for improved parameterizations. For example, Theorem~\ref{thm:delta-f-d}
gives an FPT algorithm for the parameter~$\delta_\mathcal{F}+d$. Can we
replace~$\delta_\mathcal{F}$ or~$d$ by smaller parameters? Instead of~$\delta_\mathcal{F}$, one
could consider parameters that are small when only few dependent vertices have many potential parent sets. Instead of~$d$, one could consider parameters of the directed
superstructure, for example the size of a smallest vertex cover; this parameter never
exceeds~$d$.
We think that the algorithm with running time~$3^n\cdot |I|^{\Oh(1)}$ might be practical
for~$n$ up to 20 based on experience with dynamic programs with a similar running time~\cite{KNU11}. A next step should be to combine this algorithm with heuristic data reduction and pruning rules to further increase the range of tractable values of~$n$.

\iflong
\else
\newpage
\fi

\end{document}